\documentclass[pra,aps,10pt,twocolumn,nofootinbib]{revtex4-1}

\usepackage{amsmath,amssymb,graphicx,enumerate,bbm,color}
\usepackage{amsthm}
\usepackage{hyperref}

\providecommand{\bra}[1]{\langle#1|}
\providecommand{\ket}[1]{|#1\rangle}

\newcommand{\proj}[1]{|#1\rangle\!\langle#1|}
\newcommand{\one}{\mathbbm{1}}
\newcommand{\tr}{{\rm tr}}

\newcommand{\veps}{\varepsilon}
\newcommand{\norm}[1]{\left\| #1\,\right\|_2}

\newcommand{\F}{\mathcal{F}}
\newcommand{\W}{\mathcal{W}}

\newcommand{\B}{\mathcal{B}}
\renewcommand{\S}{\mathcal{S}}

\definecolor{nblue}{rgb}{0.2,0.2,0.7}
\definecolor{ngreen}{rgb}{0.2,0.6,0.2}
\definecolor{nred}{rgb}{0.8,0.2,0.2}
\definecolor{nblack}{rgb}{0,0,0}

\newcommand{\blk}{\color{nblack}}

\newtheorem{theorem}{Theorem}
\newtheorem{definition}[theorem]{Definition}
\newtheorem{proposition}[theorem]{Proposition}

\begin{document}

\title{Imperfect measurements settings: \\ implications on quantum state tomography and entanglement witnesses} 
  
\author{Denis Rosset, Raphael Ferretti-Sch\"obitz, Jean-Daniel Bancal, Nicolas Gisin, Yeong-Cherng Liang}
\affiliation{Group of Applied Physics, University of Geneva, CH-1211 Geneva 4, Switzerland}

\date{\today}

\begin{abstract}
Reliable and well-characterized quantum resources are indispensable ingredients in quantum information processing. Typically, in a realistic characterization of these resources, apparatuses come with intrinsic uncertainties that can manifest themselves in the form of systematic errors. While systematic errors are generally accounted for through careful calibration, the effect of remaining imperfections on the characterization of quantum resources has been largely overlooked in the literature. In this paper, we investigate the effect of systematic errors that arise from imperfect alignment of measurement bases --- an error that can conceivably take place due to the limited controlability of measurement devices. We show that characterization of quantum resources using quantum state tomography or entanglement witnesses can be undermined with an amount of such imprecision that is not uncommon in laboratories. Curiously, for quantum state tomography, we find that having entanglement can help to reduce the susceptibility to this kind of error. We also briefly discuss how a given entanglement witness can be modified to incorporate the effect of such errors.
\end{abstract}

\maketitle

\section{Introduction}

The advent of quantum information science has brought inspiring opportunities for information processing~\cite{Book:MikeIke}.  At the heart of all these information processing protocols is the encoding of specific information in quantum systems and the ability to perform some specific measurements --- these features are notably important, e.g., in measurement based quantum computation~\cite{Briegel:2009}. Evidently, real life preparation of specific quantum states is never ideal and thus their reliable characterization is crucial for the implementation of these protocols.

For a complete characterization of quantum state, one uses the technique of quantum state tomography (see, e.g. Refs.~\cite{Tomography0,Thew:2002}) whereas for the purpose of entanglement verification, the technique of measuring entanglement witnesses~\cite{Terhal:2000,Lewenstein:2000,H4,Guhne:2009:PR} is widely employed. A common feature of these techniques is that they require measurements to be carried in a number of different settings. While these settings can be theoretically established easily, their experimental implementation may differ from the theoretical prescription, or come with intrinsic uncertainty, thus contributing to non-negligible systematic error. For instance, the measurement on a polarization qubit  cannot be more precise than the intrinsic uncertainty of the polarization rotator used (typically of the order of $1^\circ\sim2^\circ$ in real space\footnote{This translates to an uncertainty of $2^\circ\sim4^\circ$ on the Poincar\'e or Bloch sphere.}). Likewise, the precision of measurements on two-level atoms is limited by the effective phase and intensity uncertainties of the laser pulse experienced by the atom.

In this regard, it came as a surprise that the intrinsic uncertainty or systematic error present in measurement devices is --- to our knowledge --- hardly reported in experimental findings, clearly in stark contrast with statistical error~\cite{Christandl:1108.5329}. Also, the implication of imperfect devices seems hardly investigated beyond a relatively small number of research articles~\cite{Tomography0,Moroder:2010,Teo:1202.1713}. Of course, with careful calibration, systematic errors can usually be detected and reduced, see, e.g., Refs.~\cite{Cummins,Reichardt,Altepeter:2005,Smith:2012,Shulman:2012}.
However, it is important to note that even after careful calibration, measurement devices are after all never perfect and the intrinsic uncertainties can still manifest themselves in the form of bounded systematic errors.

The main purpose of this paper is to present concrete evidence showing that the potential implications of overlooking systematic errors can be significant in the characterization of quantum states. We illustrate this by considering a specific kind of systematic error that arises from misaligned measurement bases, and illustrating its effect on two commonly employed methods for characterizing quantum states, namely, quantum state tomography and the evaluation of entanglement witnesses~\cite{Guhne:2009:PR}. Our analysis therefore complements the approach of Ref.~\cite{Moroder:2012}, which allows the detection of systematic error from experimental data.

The paper is structured as follows. We begin in Sec.~\ref{Sec:Misalignment} by explaining the misalignment systematic error that arises from imperfect measurement settings. The notations that we are going to use in the text will be introduced in the same section. In Sec.~\ref{Sec:Tomography}, we present the effect of such systematic errors on quantum state tomography, in particular the fidelity of the reconstructed state with respect to the actual state prepared. Next, in Sec.~\ref{Sec:Witness}, we illustrate the effect of misalignment error on the evaluation of a family of genuine multipartite entanglement witnesses. In both Sec.~\ref{Sec:Tomography} and Sec.~\ref{Sec:Witness}, we also discuss how these effects can be compensated when the amount of systematic error (uncertainty) is known. We conclude in Sec.~\ref{Sec:Conclusion} with a summary of main results and some possibilities for future research. Technical details related to the main results can be found in the Appendices. 

\section{Misaligned bases from \\imperfect measurements} 
\label{Sec:Misalignment}

To study the effect of misalignment systematic errors on the characterization of quantum states, we consider the typical scenario where $n$ spatially separated qubits can each be measured locally in a number of different bases (settings). For the benefit of subsequent discussion, we remind that any of these local measurements can be described in terms of a 3-dimensional unit vector on the Bloch sphere. 
More explicitly, we shall denote by 
\begin{equation}\label{Eq:QubitObservabel}
	M_k^{(j)}=\hat{m}_k^{(j)}\cdot \vec{\sigma}
\end{equation}
the $k$-th qubit observable\footnote{Here, $\vec{\sigma}=(\sigma_x,\sigma_y,\sigma_z)$ denotes the vector of Pauli matrices.}  to be measured on the $j$-th qubit and $\hat{m}_k^{(j)}$  the corresponding Bloch vector. A misalignment  error can then be defined as follows.
\begin{definition}
A misalignment error is a systematic error that arises from imperfect measurement settings, i.e., the {\em actual} observable measured  $N_k^{(j)}=\hat{n}_k^{(j)}\cdot \vec{\sigma}$ differs from the {\em intended} measurement setting $M_k^{(j)}$ for at least some value of $j$ and $k$.
\end{definition}

Henceforth, we shall quantify the amount of error present by the quantity 
\begin{equation}\label{Eq:eps}
	\veps=\max_{j,k}{\rm acos}\left(\hat{m}_k^{(j)}\cdot\hat{n}_k^{(j)}\right).
\end{equation}
Geometrically, this means that among all the measurement settings chosen by all the $n$ parties, the {\em maximal angular deviation}\footnote{As measured in the Bloch sphere.} of the {\em actual} measurement directions from the {\em intended} ones are at most $\veps$.

To simplify the subsequent discussion, we shall also assume that all outcome probabilities can be estimated with negligible statistical error and that the actual measurement bases can always be described by $\hat{n}_k^{(j)}$ in all runs of the experiments. Next, we look into the effect of such misalignment error on some commonly employed protocols used in the characterization of a quantum state (see also Figure~\ref{Fig:Process}).

\begin{figure}[h!]
  \includegraphics[scale=1]{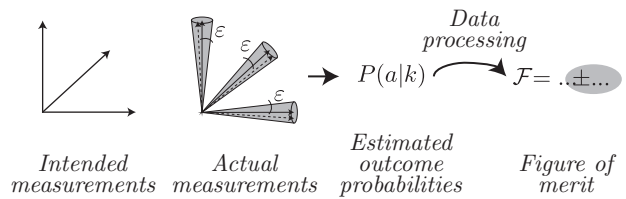}
  \caption{
    \label{Fig:Process}
A schematic representation of the procedures involved in the characterization of a quantum state and how misalignment error can affect the resulting characterization.
(1) The characterization protocol specifies measurement in a certain bases ({\em{intended} measurements}). (2) Due to imperfect measurement devices, misalignment errors creeps in during the experiments and the {\em{actual}} measurements performed differ from the intended ones. (3) The resulting measurement statistics  are used to compute a given figure of merit, such as the fidelity or the expectation value of an entanglement witness. (4) Imprecision of the measurement $\veps$  then translates into additional uncertainty in the final figure of merit.
  }
\end{figure}

\section{Implications on  quantum state tomography}
\label{Sec:Tomography}

Quantum state tomography is the process in which many copies of a quantum state are measured in a set of tomographically complete bases, followed by some state reconstruction algorithm using the measurement statistics and the {\em presumed} knowledge of the measurement bases~\cite{Tomography0,Thew:2002,Teo:1202.1713}. In this section, we illustrate the effect of misalignment systematic error on the tomography of $n$-partite qubit states.

Throughout, we shall assume that the qubit tomography is {\em intended} to be carried out in the standard Pauli bases, i.e.,
\begin{equation}\label{Eq:StandardBasis}
	M_1^{(j)}=\sigma_x, \quad M_2^{(j)}=\sigma_y,\quad M_3^{(j)}=\sigma_z,
\end{equation}
for all parties. Moreover, we shall quantify the effect of misalignment errors on quantum state tomography using the Uhlmann-Jozsa fidelity~\cite{Uhlmann:1976transition,Jozsa:1994fidelity} between the actual state $\tau$ and the reconstructed state $\rho$, i.e., 
\begin{equation}\label{Eq:fid}
    \F(\tau,\rho)= \left(\tr\sqrt{\sqrt{\tau}\rho \sqrt{\tau}}\right)^2.
\end{equation}
It is worth noting that when either $\rho$ or $\tau$ is a pure state, the expression above reduces to $\F(\tau,\rho)=\tr (\rho \, \tau)$. 

Clearly, the smaller is the value of $\mathcal{F}(\tau,\rho)$, the more drastic is the effect of misalignment error on quantum state tomography.

\subsection{Single qubit state tomography}
\label{Sec:Tomography:SingleQubit}

Let us begin with the simplest example of a single-qubit state tomography. The pedagogical example given below will also serve to remind the key features involved in some of the standard state reconstruction techniques, such as {\em linear inversion} and {\em maximum-likelihood estimation}.

\subsubsection{A simple example of erroneous state reconstruction starting from a pure state}
\label{Sec:Tomography:SingleQubit:Example}

 Consider a source that produces a quantum state $\tau$ as parametrized by the Bloch vector $\vec{t}$:
\begin{equation}
\label{Eq:Tomography:ActualState}
	\tau = \frac{\one + \vec{t} \cdot \vec{\sigma}}{2}.
\end{equation}
Suppose now that a qubit state tomography is to be carried out for this source via the {\em intended measurements}\footnote{For simplicity, we omit all superscripts in the single-qubit scenario.}
\begin{equation}\label{Eq:Single:StandardBasis}
	M_1=\sigma_x, \quad M_2=\sigma_y,\quad M_3=\sigma_z,
\end{equation}
whereas in reality, due to the presence of misalignment errors, the actual observables measured are described instead by $\{N_k\}$. 

From Born's rule, we can compute the outcome probability for the $k$-th measurement setting as:
\begin{equation}
\label{Eq:Tomography:Outcomes}
	P(\pm1|k)=\tr\Bigg(\tau\,\frac{\one\pm\,N_k}{2}\Bigg), \quad \text{for } k=1,2,3.
\end{equation}
The essence of state reconstruction is to find a legitimate density matrix $\rho$, referred as the {\em reconstructed state} such that
\begin{equation}\label{Eq:Tomography:Reconstruction}
	P(\pm1|k)=\tr\Bigg(\rho\,\frac{\one\pm\,M_k}{2}\Bigg).
\end{equation}
Since this amounts to solving a set of equations that are linear in the measurement statistics, this procedure of solving for the reconstructed state $\rho$ is also known as {\em linear inversion}.
Note that the reconstruction is done using the ideal description given in Eq.~\eqref{Eq:Single:StandardBasis}. Evidently, if the actual misalignment error was detected, the reconstruction procedure could be corrected by replacing $\{ M_k \}$ with $\{N_k \}$ in Eq.~\eqref{Eq:Tomography:Reconstruction}. 

Specifically, imagine that the actual state prepared is $\tau = \proj{\psi_s^+}$, where $\ket{\psi_s^+}$ is the positive eigenstate of $\hat{s} \cdot \vec{\sigma}$ with $\hat{s} = \tfrac{1}{\sqrt{3}}(1,1,1)^T = \vec{t}$:
\begin{equation}
\label{Eq:psi_s_pm}
	\ket{\psi_s^\pm}= \frac{1}{\sqrt{3\mp\sqrt{3}}} 
	\left [ \ket{0} \pm \sqrt{2\mp\sqrt{3}}e^{i \frac{\pi}{4}} \ket{1}\right], 
\end{equation}
 and the actual tomography measurement directions diverge uniformly from the intended directions as $\veps$ increases (see Figure~\ref{Fig:Tomography:Axes}), i.e.,
\begin{equation}
  \label{Eq:Tomography:ExampleMisalignement}
  \hat{n}_1 = \left(
    \begin{array}{r}
      c_\varepsilon \\
      -\tfrac{s_\varepsilon}{\sqrt{2}} \\
      -\tfrac{s_\varepsilon}{\sqrt{2}}
    \end{array}
  \right),
  \hat{n}_2 = \left(
    \begin{array}{r}
      -\tfrac{s_\varepsilon}{\sqrt{2}}\\
      c_\varepsilon \\
      -\tfrac{s_\varepsilon}{\sqrt{2}}
    \end{array}
  \right),
  \hat{n}_3 = \left(
    \begin{array}{r}
      -\tfrac{s_\varepsilon}{\sqrt{2}} \\
      -\tfrac{s_\varepsilon}{\sqrt{2}} \\
      c_\varepsilon 
    \end{array}
  \right),
\end{equation}
where
\begin{equation}\label{Eq:cs:eps} 
	c_\varepsilon = \cos \varepsilon, \quad s_\varepsilon = \sin \varepsilon.
\end{equation}

\begin{figure}
  \includegraphics[scale=.8]{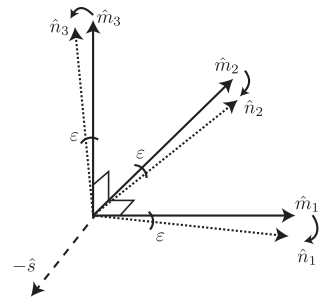}
  \caption{
    \label{Fig:Tomography:Axes}
    Intended and actual measurement directions for
    the tomography of $\ket{\psi_s}$. The actual measurement directions $\{\hat{n}_k\}$ given in Eq.~\eqref{Eq:Tomography:ExampleMisalignement} correspond to the triad of $\{\hat{m}_k\}$ ``opened" uniformly towards $-\hat{s}=-\tfrac{1}{\sqrt{3}}(1,1,1)$ while satisfying  ${\rm acos}(\hat{m}_k\cdot \hat{n}_k)= \veps$.}
\end{figure}

It now follows from Eq.~\eqref{Eq:Tomography:Outcomes}--Eq.~\eqref{Eq:Tomography:ExampleMisalignement} that
\begin{equation}
  \vec{r} = \left(\begin{array}{ccc}
      c_\varepsilon & - \tfrac{s_\varepsilon}{\sqrt{2}} & - \tfrac{s_\varepsilon}{\sqrt{2}}\\
      - \tfrac{s_\varepsilon}{\sqrt{2}} & c_\varepsilon & - \tfrac{s_\varepsilon}{\sqrt{2}}\\
      - \tfrac{s_\varepsilon}{\sqrt{2}} & - \tfrac{s_\varepsilon}{\sqrt{2}} & c_\varepsilon
    \end{array}\right) \hat{s}=(c_\varepsilon - \sqrt{2}s_\veps)\hat{s}.
  \label{Eq:Tomography:LinearEffect}
\end{equation}
The Bloch vector $\vec{t}=\hat{s}$ of $\tau = \proj{\psi_s^+}$ is thus an
eigenvector of the above linear transformation, with eigenvalue $c_\varepsilon - \sqrt{2}
s_\varepsilon$. Hence, as long as $|c_\varepsilon - \sqrt{2}s_\varepsilon|\le 1$, which takes place for 
$0\le\veps\le\text{acos}\frac{1}{3}\approx 70^\circ$, the reconstructed state obtained by solving Eq.~\eqref{Eq:Tomography:Reconstruction} is always a legitimate quantum state. It is then straightforward to verify that $\rho$ can be written as a convex mixture:
\begin{equation}
  \rho = f \proj{\psi_s^+} + \left( 1 -f \right) \proj{\psi_s^-},
\end{equation}
where
\begin{equation}
  \label{Eq:Tomography:FidelityLoss}
  f(\veps)=\F(\tau,\rho)=\frac{1}{2}\left (1+\cos \veps -\sqrt{2} \sin \veps \right )
\end{equation}
is simply the fidelity of the reconstructed state $\rho$ with respect to the actual state $\tau= \proj{\psi_s^+}$ [cf. Eq.~\eqref{Eq:fid}].

This implies, for instance, that  with a $2^\circ$ misalignment in $\hat{m}_k$ but everything else perfect, the erroneously reconstructed state still only has 97.5\% fidelity with respect to the actual state (Figure~\ref{Fig:Tomography:Fidelity}). In this regard, note that Eq.~\eqref{Eq:Tomography:FidelityLoss} actually represents the worst-case fidelity for any actual state $\tau$ that is pure and where the intended tomographic measurements are given by Eq.~\eqref{Eq:Single:StandardBasis}. The proof of this is somewhat involved and is relegated to Appendix~\ref{App:Mixed1Qubit}. More generally, to study the effect on fidelity for small $\veps$, we shall introduce the notion of {\em susceptibility} to misalignment errors, defined as:
\begin{equation}
  \label{Eq:Tomography:Susceptibility}
  \mathcal{S} = \frac{\partial f(\veps)}{\partial \veps} \bigg |_{\veps\rightarrow0}=-\frac{1}{\sqrt{2}}.
\end{equation}
In particular $\mathcal{S}\neq0$ shows that misalignment errors have a first-order effect on the fidelity. 

\begin{figure}
  \resizebox{7cm}{5cm}{\includegraphics{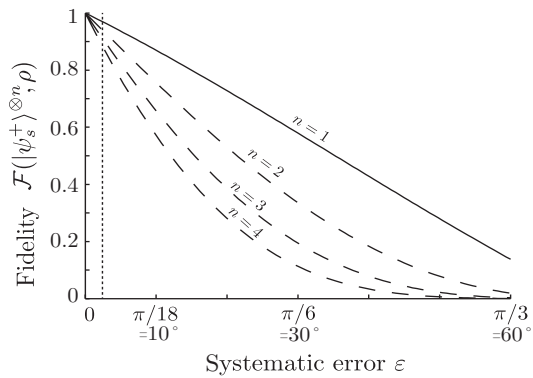}}
  \caption{
    \label{Fig:Tomography:Fidelity}
    Fidelity of the reconstructed state with respect to the initial state $\ket{\psi_s^+}^{\otimes n}$ when the actual measurement settings are those described in Figure~\ref{Fig:Tomography:Axes} and the intended measurement bases are those given in Eq.~\eqref{Eq:StandardBasis}. The vertical dotted line corresponds to $\veps=2^\circ$ as 
discussed in the text. }
\end{figure}

\subsubsection{Effect on state reconstruction starting from a general qubit state}
\label{Sec:Tomography:SingleQubit:WorstLoss}

Obviously, in a realistic experimental situation, we do not expect any source to produce a pure qubit state. A relevant problem to determine is thus whether the expression given in Eq.~\eqref{Eq:Tomography:FidelityLoss} still represents the worst-case fidelity even in this more general scenario. 

Before answering the above question, it is important to understand that the linear inversion process described in Sec.~\ref{Sec:Tomography:SingleQubit:Example} does not always lead to a physical state $\rho$. For instance,  in the example given above, if $\veps>\text{acos}\tfrac{1}{3}$, the reconstructed Bloch vector as given by Eq.~\eqref{Eq:Tomography:LinearEffect} would have length greater than 1, and thus corresponds to an unphysical state $\rho$.

A commonly employed technique to circumvent this kind of problem is to make use of the maximum-likelihood estimation (MLE) technique introduced by Hradil~\cite{Hradil:1997}. From the measurement statistics and the supposed knowledge of the measurement bases, this technique seeks to find a physical quantum state that maximizes the (log) likelihood function, and hence determines the quantum state that is most likely to give rise to the experimental data.

For a general mixed qubit state and a set of 3 measurement directions $\{\hat{n}_k\}$ satisfying Eq.~\eqref{Eq:eps}, one does not always obtain a physical state via linear inversion. Nonetheless, we prove in Appendix~\ref{App:Mixed1Qubit} that with MLE, the fidelity of the reconstructed state cannot be worse than that given in Eq.~\eqref{Eq:Tomography:FidelityLoss}, i.e., 
\begin{equation}\label{Eq:Tomography:SingleQubit:Bound}
	\mathcal{F}(\tau',\rho') \ge f(\veps),
\end{equation}
for any qubit state $\tau'$ and the corresponding state $\rho'$ reconstructed from the MLE algorithm. 

\subsection{Multiqubit state tomography}

Let us now study the effect of misalignment error on the tomography of  multiqubit states. As we will see below, in the reconstruction of multiqubit states, entanglement also plays a nontrivial role in combating the effect of misalignment errors.

\subsubsection{Multiqubit product states}

To start off, note that the simple one-qubit example given above can be easily generalized to the $n$-partite scenario if the source actually produces a product state, i.e., 
\begin{equation}
	\tau=\bigotimes_{j=1}^n \tau^{(j)}. 
\end{equation}
To see this, we note that the product nature of quantum states is preserved by the MLE state reconstruction technique (for a proof of this, see, e.g., Appendix~\ref{App:Proof}).  Moreover, the reconstruction procedure can be carried out independently for each qubit. Thus, if we define analogously $f_n(\veps)$ the worse-case fidelity in the $n$-partite case, it follows that
\begin{equation}
  \label{Eq:Tomography:NQubits:Fidelity}
  \F\left ( \bigotimes_{j=1}^n \tau^{(j)},\rho \right ) \ge
  f_n(\varepsilon) = \left [ f(\varepsilon) \right ]^n,
\end{equation}
with susceptibility 
\begin{equation}
	\S = \frac{\partial f_n}{\partial \veps}|_{\veps\rightarrow 0}=-\frac{n}{\sqrt{2}}.
\end{equation}
Note that  the inequality in Eq.~\eqref{Eq:Tomography:NQubits:Fidelity} is saturated by considering $\tau=\proj{\psi_s^+}^{\otimes n}$ and where all its constituents are measured with axes defined in Eq.~\eqref{Eq:Tomography:ExampleMisalignement} (see also Figure~\ref{Fig:Tomography:Axes}). Clearly, this shows that the effect of misalignment error may accumulate  with the number of parties. The fidelity itself $f_n(\veps)$ as a function of $\veps$ for $n\le 4$ is plotted in Figure~\ref{Fig:Tomography:Fidelity}.

\subsubsection{Two-qubit-entangled states}
\label{Sec:Tomography:TwoQubits:Pure}

Evidently, in the context of quantum information processing, it is arguably more relevant to look into the robustness of entangled states with respect to the aforementioned systematic errors. To this end, we have performed numerical optimization to determine --- for small $\veps$ and for fixed amount of entanglement (as parametrized by $\alpha\in\left[0,\frac{\pi}{4}\right]$) --- the {\em worst-case} fidelity $\F(\ket{\psi_\alpha},\rho)$ by varying over the misaligned measurement settings (as parameterized by $\hat{n}_k^{(j)})$ and arbitrary qubit basis states $\ket{\psi^\pm_j}$ in
\begin{equation}
	\ket{\psi_\alpha}=\cos\alpha \ket{\psi^+_1} \ket{\psi^+_2}+\sin\alpha\ket{\psi^-_1} \ket{\psi^-_2}.
\end{equation}

In our optimization,\footnote{To make the optimization more efficient and robust, we have also provided the gradient of the objection function (with respect to the parameters of the problem) to the optimization solver. } 
we assume Eq.~\eqref{Eq:StandardBasis} and focus on small error,  namely, $\veps\le\frac{\pi}{200}$ to determine the worst-case fidelity $\F(\ket{\psi_\alpha},\rho)$ numerically for $\veps$ in this domain. We then estimate numerically the susceptibility, i.e., the initial slope $\S\left( \alpha \right) =  \nabla_\veps \mathcal{F} (\ket{\psi_\alpha},\rho)|_{\veps=0}$. The results are shown in~Figure.~\ref{fig_twoqubit_tomography}. Interestingly, our results show that in the worst-case scenario, pure product states are the least robust against systematic errors that arise from misaligned measurements. In the two-qubit case ($n=2$), Eq.~\eqref{Eq:Tomography:NQubits:Fidelity} thus provides the worst-case rate of decrease of fidelity with respect to $\veps$ for small  $\veps$. It is also interesting to note that maximally entangled two-qubit pure state does not appear to be the most robust against this kind of error.
  
  \begin{figure}[h]
    \resizebox{7cm}{!}{\includegraphics{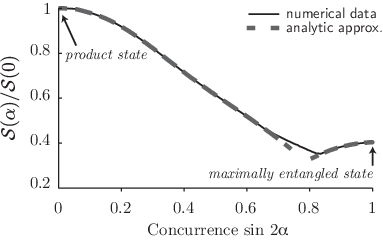}}
    \caption{\label{fig_twoqubit_tomography}
      Susceptibility of 2-qubit pure state $\ket{\psi_\alpha}$ to misalignment 
      error $\veps$ as a function of the entanglement present in $\ket{\psi_\alpha}$
       (parameterized by the concurrence~\cite{W.K.Wootters:PRL:2245} of $\ket{\psi_\alpha}$). 
      For given $\alpha$, $\S(\alpha)$ gives the rate of decrease of the fidelity with respect to $\veps$ as $\veps\to0$; 
      $\S(0)$ is the corresponding intial slope for pure product state. The 257 numerical data points obtained from 
      $4\times10^4$ \blk optimizations are plotted in a solid line.  The dashed lines represent segments
      of the plot that can be very well approximated using the explicit
      parameterizations given in Appendix~\ref{App:TwoQubit}. }
  \end{figure}

What gives entangled state more resistance to this kind of systematic error in the worst-case scenario? Our intuition is that {\em uncorrelated}, {\em local} misalignment errors have mostly local effect. Here, the misalignment errors  are {\em uncorrelated} in the sense that for each intended measurement direction, say, for the second party $\hat{m}_{k_2}^{(2)}$, its actual, deviated measurement direction $\hat{n}_{k_2}^{(2)}$ is independent of the choice of measurement $k_1$ by the first party. To verify this intuition, we have performed similar analysis allowing the actual measurement direction to vary depending on the choice of measurement of the other party. Indeed, it turns out that pure product state is no longer the most fragile one against misalignment error in this more general scenario. More details on this analysis can be found in Appendix~\ref{App:CorrelatedError}.

Coming back to the uncorrelated case, we note that for small amount of entanglement, say, $\alpha\le\tfrac{3\pi}{32}$, the reconstructed state $\rho$ loses its fidelity with respect to the actual state $\ket{\psi_\alpha}$    --- as quantified by $1-\F(\ket{\psi_\alpha},\rho)$ --- predominantly via terms that are proportional to the length of the Bloch vector of the reduced density matrix. Since this length $\cos 2\alpha$  {\em shrinks} as $\alpha$ increases, clearly, among all the weakly entangled two-qubit pure states, the pure product state is  the most  susceptible to this kind of systematic error. A more formal analysis of this is given in Appendix~\ref{App:TwoQubit} .

\section{Implications on entanglement certification}
\label{Sec:Witness}

While standard quantum state tomography can be carried out for a system involving a small number of qubits, in the realms where quantum information processing is advantageous against its classical counterpart, this complete characterization is practically infeasible (see, however, Ref.~\cite{Tomography}). Next, we shall look at the implication of misaligned measurements on partial characterization of  quantum state via entanglement witness. 

\subsection{Preliminaries}

A witness $\W$ for genuine $n$-partite entanglement is a Hermitian observable that satisfies
\begin{equation}
  \label{Eq:EntanglementWitness:Definition}
	 \tr\,\left( \W\,\rho_\text{bisep.}\right)\ge 0,
\end{equation}
for all {\em biseparable states} $\rho_\text{(bi-)sep.}$ but is violated by at least some (genuinely $n$-partite) entangled states~\cite{H4,Guhne:2009:PR}. An {\em optimized} entanglement witness $\W$, moreover, satisfies the property that there must exist biseparable quantum state $\rho_{\rm bisep.}$ such that the defining inequality, cf. Eq.~\eqref{Eq:EntanglementWitness:Definition} is saturated. Geometrically, this means that the separating hyperplane defined by $\W$ is actually tangential to the set of biseparable states (see Figure~\ref{Fig:EntanglementWitness:Drawing}).

\begin{figure}[h!]
  \includegraphics{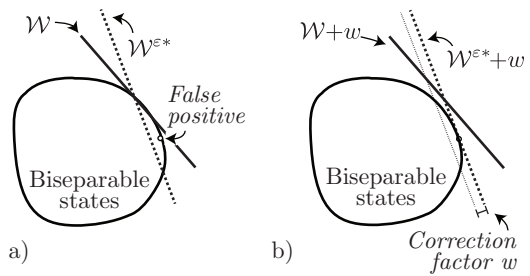}
  \caption{
    \label{Fig:EntanglementWitness:Drawing} 
    A schematic diagram showing the effect of misalignment error on the evaluation of an entanglement witness.
    An optimized entanglement witness $\mathcal{W}$ is tangent to the boundary of the set of biseparable states. (a) When evaluated using misaligned measurements, the witness can cross the boundary (as $\mathcal{W}^\veps$) and thus some biseparable states may appear to be genuinely $n$-partite entangled. (b) To correct the problem, one can evaluate the correction factor $w(\veps)$ such that no biseparable state give a false positive result.
  }
\end{figure}

Though being more economical in terms of resource requirements, we shall demonstrate below that entanglement certification via (optimized) entanglement witness is relatively  fragile against misalignment systematic errors (see also Figure~\ref{Fig:EntanglementWitness:Drawing}). For definiteness, we assume in subsequent analysis that a linear entanglement witness $\W$ is {\em intended} to be evaluated by measuring {\em local} observables $M^{(j)}_{k_j}=\hat{m}^{(j)}_{k_j}\cdot\vec{\sigma}$. And as with the rest of the paper, we assume that due to the presence of misalignment errors, the actual local observable measured is $N^{(j)}_{k_j} = \hat{n}^{(j)}_{k_j} \cdot \vec{\sigma}$, where the angular deviation of $\hat{n}^{(j)}_{k_j}$ from $\hat{m}^{(j)}_{k_j}$ is bounded by $\veps$, cf.  Eq.~\eqref{Eq:eps}.

To incorporate the effect of uncorrelated misalignment error, one can first determine the correction factor
\begin{equation}
  \label{Eq:EntanglementWitness:Correction}
  w(\veps)=\min_{\W^\veps}\min_{\rho_{\rm bisep.}} \tr\,\left( \W^\veps\,\rho_{\rm bisep.}\right),
\end{equation}
where the minimization of $\W^\veps$ is to be carried out over {\em all possible} Hermitian observables $\W^\veps$ satisfying the constraint given in Eq.~\eqref{Eq:eps}. We write $\W^{\veps*}$ the Hermitian observable giving the minimal value of $w(\veps)$.\footnote{Note that by convexity of the set of biseparable states, it suffices to consider pure biseparable quantum state $\ket{\Psi_{\rm bisep}}$ in the minimization of Eq.~\eqref{Eq:EntanglementWitness:Correction}.}

The function $w(\veps)$ thus gives the worst-case value of the witness $\W$ with respect to all biseparable states in the presence of bounded misalignment error $\veps$. If $\veps$ is known, the witness $\W$ can then be modified in the following way 
\begin{equation}
  \W\to \W'=\W-w(\veps)\one^{\otimes n}
\end{equation}
such that
\begin{equation}
  \label{Eq:EntanglementWitness:Definition1}
  \tr\,\left( \W'\,\rho_{\rm bisep.}\right)\ge 0,
\end{equation}
holds true for all biseparable states $\rho_{\rm bisep.}$ even if we allow misalignment error bounded by $\veps$, see Figure~\eqref{Fig:EntanglementWitness:Drawing}.

\subsection{A bipartite entanglement witness and its correction factor in the presence of bounded misalignment error}

Let us now look at some explicit examples. Consider the following two-qubit entanglement witness constructed from the singlet state $\ket{\Psi^-}$, 
\begin{equation}
  \label{Eq:EntanglementWitness:Singlet}
  \W_{\Psi^-} = \frac{1}{2}\one^{\otimes 2}-\proj{\Psi^-} =\frac{1}{4}\one^{\otimes 2}+\frac{1}{4}\sum_{k=x,y,z} \sigma_k\otimes \sigma_k,
\end{equation}
where $\one$ is the $2\times2$ identity matrix. Clearly, a natural way to evaluate this witness experimentally involves measurements in the Pauli bases, i.e., with $M_k^{(j)}$ given by Eq.~\eqref{Eq:StandardBasis}.

Imagine now  a physical system prepared in the separable state:
\begin{equation}
  \ket{\psi}=\cos^2\chi\left(\ket{0}+e^{i\frac{\pi}{4}}\,\tan\chi\ket{1}\right)\left(\tan\chi\ket{0}+ e^{-i\frac{3\pi}{4}}\ket{1}\right),
\end{equation}
where $\chi=\tfrac{{\rm asec}\sqrt{3}}{2}$ and, instead of the Pauli bases, measurements were made --- due to misaligned measurements --- along the following directions on the Bloch sphere
\begin{equation}
\label{Eq:Tomography:ExampleUnphysical}
\hat{n}_1^{(1,2)} = \left(
\begin{array}{r}
c_\varepsilon \\
\tfrac{s_\varepsilon}{\sqrt{2}} \\
\tfrac{s_\varepsilon}{\sqrt{2}}
\end{array}
\right),\quad
\hat{n}_2^{(1,2)} = \left(
\begin{array}{r}
\tfrac{s_\varepsilon}{\sqrt{2}}\\
c_\varepsilon \\
\tfrac{s_\varepsilon}{\sqrt{2}}
\end{array}
\right),\quad
\hat{n}_3^{(1,2)} = \left(
\begin{array}{r}
\tfrac{s_\varepsilon}{\sqrt{2}} \\
\tfrac{s_\varepsilon}{\sqrt{2}} \\
c_\varepsilon 
\end{array}
\right).
\end{equation}

An intended measurement on $\W_{\Psi^-}$ using $\ket{\psi}$ therefore results in the measurement of 
\begin{equation}
	\W_{\Psi^-}^\veps=\frac{1}{4}\one^{\otimes2}+ \frac{1}{4}\sum_{k=1}^3 \hat{n}_k^{(1)}\cdot\vec{\sigma}\otimes\hat{n}_k^{(2)}\cdot\vec{\sigma},
\end{equation} 
which gives an expectation value of
\begin{equation}\label{Eq:CorrectionSinglet}
  \bra{\psi}W_{\Psi^-}^\veps\ket{\psi} = \frac{1}{8}\left(\cos\,2\veps -2\sqrt{2}\sin\,2\veps -1\right),
\end{equation}
which is negative for all $0<\veps<\frac{\pi}{2}$.
In other words, as soon as $\veps>0$, an evaluation of the above entanglement witness $\W^\veps_{\Psi^-}$ using the separable state $\ket{\psi}$ will always give an affirmative, but erroneous certification that the state is entangled. Numerically, the above strategy also corresponds to the minimal value that we have found for the optimization specified in Eq.~\eqref{Eq:EntanglementWitness:Correction}. Therefore, for the witness $\W_{\Psi^-}$, our result  suggests that the correction factor $w(\veps)$ is given by Eq.~\eqref{Eq:CorrectionSinglet}.

\subsection{A witness for genuine $n$-qubit entanglement and its correction factor in the presence of bounded misalignment error}

Likewise, for the detection of genuine multipartite entanglement, let us consider the following $n$-partite entanglement witness~\cite{Bourennane:2004}:
\begin{equation}
  \W_\text{GHZ} = \frac{1}{2}\one^{\otimes n} - \proj{\text{GHZ}},
\end{equation}
where $\ket{\text{GHZ}}=\tfrac{1}{\sqrt{2}}(\ket{0}^{\otimes n}+\ket{1}^{\otimes n}$ is the $n$-partite Greenberger-Horne-Zeilinger state. An economical way to measure this witness is to have all the $n$ parties performing the same measurements~\cite{Guhne:2007wd}, i.e., 
\begin{equation}\label{Eq:Mkj}
	M_k^{(j)}=\cos\frac{k\pi}{n}\sigma_x+\sin\frac{k\pi}{n}\sigma_y, \quad\text{for } k=1,\ldots,n,
\end{equation}
and $M_{n+1}^{(j)}=\sigma_z$. The measurement statistics on these settings can then be combined to give the desired expectation value of $\W_\text{GHZ}$ in the following way:
\begin{equation}\label{Eq:WitnessGHZ}
  \W_{\ket{\text{GHZ}}}\!=\!\frac{1}{2}\bigg[\!\one^{\otimes n} -\!\! \sum_{\ell=\pm1}\!\!\!\bigg(\frac{\one+\ell\,\sigma_z}{2}\bigg)^{\!\!\!\otimes n}
  \!\!-\sum_{k=1}^n \frac{(-1)^k}{n} \bigotimes_{j=1}^n M_k^{(j)}\!
  \bigg].
\end{equation}
In what follows, we provide estimates of the correction factor $w_\text{GHZ}(\veps)$ of $\W_{\ket{\text{GHZ}}}$ --- as determined by numerical optimization --- for $n\ge 4$, separating the cases of $n$ even and $n$ odd. This correction factor $w_\text{GHZ}(\veps)$ is plotted for $n=3,\ldots,8$ in Figure~\ref{Fig:EntanglementWitness:Value}.
\begin{figure}
  \includegraphics[scale=1.1]{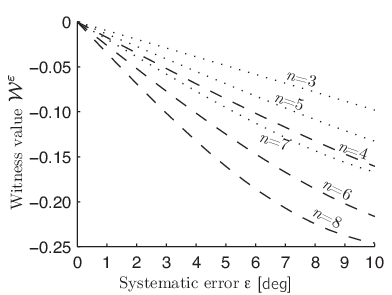}
  \caption{
    \label{Fig:EntanglementWitness:Value}
    The worst expectation value of $\W^\veps$ found using bispearable states and systematic error $\veps\le10^\circ$. When $\veps=0$, $\W^\veps$ reduces to $\W_{\ket{\text{GHZ}}}$ given in Eq.~\eqref{Eq:WitnessGHZ}. The curves for $n\ge 4$ are computed using Eqs.~\eqref{Eq:CorFac:Even} and~\eqref{Eq:CorFac:Odd}, whereas the curve for $n=3$ has been obtained numerically.}
\end{figure}

\subsubsection{Estimated correction factor for even $n\ge 4$}

For $\W_{\ket{\text{GHZ}}}$ with even $n\ge4$ and $\veps \le \tfrac{\pi}{2n}$, numerical optimizations suggest that the correction factor is given by:
 \begin{equation}\label{Eq:CorFac:Even}
   w_\text{GHZ}^{n\text{ even}}(\veps)=-\frac{1}{4}\sin\,n\veps.
 \end{equation}
This can be achieved by considering the $n$-qubit biseparable state
\begin{equation}
 	\ket{\psi_n}=\frac{1}{2}(\ket{0}^{\otimes \frac{n}{2}}+e^{i\frac{\pi}{4}}\ket{1}^{\otimes \frac{n}{2}})\otimes(\ket{0}^{\otimes \frac{n}{2}}+e^{-i\frac{\pi}{4}}\ket{1}^{\otimes \frac{n}{2}}),
 \end{equation}
for even $n\ge4$, with the following misaligned observables:
 \begin{equation} 
 	N_k^{(j\le\frac{n}{2})}=\cos\left[\frac{k\pi}{n}+(-1)^k\veps\right]\sigma_x+\sin\left[\frac{k\pi}{n}+(-1)^k\veps\right]\sigma_y,
 \end{equation}
 and
 \begin{equation} 
 	N_k^{(j>\frac{n}{2})}=\cos\left[\frac{k\pi}{n}-(-1)^{k}\veps\right]\sigma_x+\sin\left[\frac{k\pi}{n}-(-1)^{k}\veps\right]\sigma_y
 \end{equation}
 (see Figure~\ref{Fig:EntanglementWitness:EvenAxes}). Clearly, Eq.~\eqref{Eq:CorFac:Even} is negative as soon as  $\veps>0$. Thus, as with the two-qubit example given above, measuring  the biseparable state $\ket{\psi_n}$ for these imperfectly implemented witnesses always results in an erroneous certification of the non-separability of the state.
 \begin{figure}
   \includegraphics[scale=0.8]{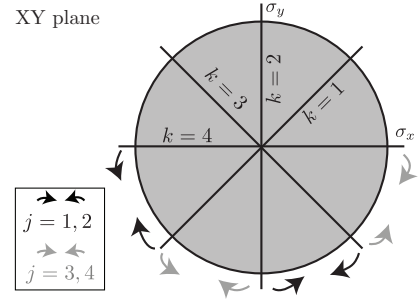}
   \caption{
     \label{Fig:EntanglementWitness:EvenAxes} 
     Evaluating the entanglement witness $\W_\text{GHZ}$ with the intended measurement settings $M_k^{(j)}$ on the XY plane (drawn here for $n=4$ parties). These settings correspond to equally spaced measurement directions $\hat{m}_k^{(j)}$ on the XY plane, with neighboring ones separated by an angle $\tfrac{k\pi}{n}$. The systematic errors considered here correspond to having these rays moving in pairs towards each other as $\veps$ increases. Measurement directions are deviated according to the black arrows for qubits numbered $j=1,2$ and to the grey arrows for $j=3,4$. Measurement in the $\sigma_z$ basis is assumed to stay unperturbed.}
 \end{figure}
 
 \label{App:WitnessNQubit}
 \subsubsection{Estimated correction factor for odd $n\ge 5$}

 For $\W_{\ket{\text{GHZ}}}$ with odd $n\ge5$ and $\veps \lesssim \tfrac{\pi}{2n}$, numerical optimizations suggest that the correction factor is given by:
 \begin{equation}
\label{Eq:CorFac:Odd}
   w_\text{GHZ}^{n\text{ odd}}(\veps)=\frac{1}{4n}\bigg [ n-2 -  (n-1)\cos\,\veps+\cos\,n\veps-\frac{\sin\,n\veps}{\tan \frac{\pi}{2n}} \bigg ].
 \end{equation}
 This can be achieved using the biseparable pure state $\ket{\psi_{n}}=\ket{\psi_{-}}\otimes\ket{\psi_{+}}$,
 \begin{gather}
   \ket{\psi_{-}}=\frac{1}{\sqrt{2}}\left(\ket{0}^{\otimes  n_-} 
     + e^{- \nu i \frac{\left( 3 n + \nu \right) \pi}{4 n}} \ket{1}^{\otimes n_-}\right),\\
   \ket{\psi_{+}}=\frac{1}{\sqrt{2}}\left(\ket{0}^{\otimes  n_+} 
     + e^{\nu i \frac{\left( 3 n + \nu \right) \pi}{4 n}}  \ket{1}^{\otimes n_+}\right),
 \end{gather}
where $n_{\pm}=\frac{n\pm1}{2}$ and $\nu = (-1)^{n_-}$. We choose the following local measurement settings for $k=1,\ldots, n$:
 \begin{equation}
   \begin{split}
     N_k^{(j)}=\cos\left[\frac{k\pi}{n}+ g_{k}^{(j)}\veps\right]\sigma_x
     +\sin\left[\frac{k\pi}{n}+g_{k}^{(j)}\veps\right]\sigma_y,
   \end{split}
 \end{equation}
keeping $N_{n+1}^{(j)}=M_{n+1}^{(j)}=\sigma_z$ and defining $g_{k}^{(j)}= \nu (-1)^k \text{sign } [(k-n_+)(j-\tfrac{n}{2})]$.

 \section{Conclusion}\label{Sec:Conclusion}
 
 Intrinsic uncertainties in measurement devices, which can manifest themselves in the form of  misalignment systematic errors represent an unavoidable part of any real-life quantum experiment. In this paper, we show by explicit examples that the procedure of characterizing quantum resources using state tomography or entanglement witnesses can be considerably affected when such systematic errors are not properly taken care of. For example, when considering pure two-qubit state tomography, every single degree of misalignment on the Bloch sphere can potentially lead to $\approx1\%$ decrease in the fidelity of the reconstructed state. For general product state, the worst loss of fidelity has a scaling that is linear in $n$, making them increasingly sensitive to such systematic error.
 
For the verification of entanglement via an optimized entanglement witness, we show that an erroneous certification could arise {\em whenever} there is nonzero misalignment error.  While our demonstration was made for specific entanglement witnesses, it should be emphasized that {\em all} non-device-independent entanglement witnesses~\cite{diew} are potentially susceptible to this kind of imperfection, and the procedure we followed can also be applied to them. But all is not lost, the effect of misalignment error, as we have demonstrated, can be incorporated by modifying a given entanglement witness. In this regard, it could be interesting to understand the amount of potential misalignment systematic error present in some of the state-of-the-art characterizations of quantum resources, such as those in Refs.~\cite{superconducting,ions,photons}.  Alternatively,  entanglement verification can also be carried out without such characterization by implementing device-independent entanglement witnesses~\cite{diew,diewTV,GUBI} provided by Bell-like inequalities. 

Let us now comment on some further possibilities for future research. Firstly, for quantum state tomography, our analysis of entangled two-qubit states focused on pure states; based on the numerical optimizations that we have done, we conjecture that the bound given in Eq. (18) holds for mixed two-qubit states as well. Obviously, similar studies for n-qubit systems should be carried out for n ≥ 3. The increased resistance that we have observed in entangled two-qubit states against misalignment error suggests that entanglement also plays a nontrivial role in quantum state estimation — something that deserves to be understood better. For instance, it would be interesting to see what role entanglement
plays when considering other imperfections, such as mismatched detector efficiencies. We remind also that our analysis on the effect of misaligned measurements was carried out at an abstract level where, in particular, each measurement basis can be misaligned differently but in an uncorrelated manner (see, however, Appendix A4). In practice, typical errors present in particular experimental setups may be more/less general than considered here, leading to larger/smaller effects. It would thus be interesting to adapt the analysis that we have presented here for some actual physical system (e.g. superconducting qubits, qubits in ion traps) and see how the effect changes.

Clearly, our work only marks the beginning of a deeper understanding how imperfect devices can affect real-life characterization of quantum resources. In the long run, it is clearly desirable to develop a general method for computing the additional uncertainty that should be incorporated in any figure of merit as a result of any given imprecision in the measurement device.  The joint effect of imperfect devices and finite statistical error is evidently also a relevant question that needs to be addressed.
 
Of course, it is also of general interest to understand how imperfect measurement settings directly affect quantum information processing tasks, which evidently require more than well-characterized quantum resource. To this end, we note that the effect of imperfect measurement settings on measure and prepare quantum key distribution protocol is investigated in parallel in Ref.~\cite{Woodhead:2012}. 
 
 \paragraph*{Acknowledgments.}
 We acknowledge useful discussions with Stefano Pironio and Clara Osorio. This work was supported by the Swiss NCCRs QP and QSIT, the CHIST-ERA DIQIP and the European ERC-AG QORE.

 \appendix
 
 
\section{Miscellaneous details related to quantum state tomography}
\label{App:Tomography}


 \subsection{Tomography of a single qubit mixed state}
 \label{App:Mixed1Qubit}
 
Under the assumption that the misalignment systematic error is upper-bounded by $\veps$, we show in this Appendix that the minimal fidelity of the reconstructed single-qubit state $\rho$ with respect to the actual state $\tau$, i.e., $\F(\tau,\rho)$ is indeed given by Eq.~\eqref{Eq:Tomography:FidelityLoss}.  Throughout this Appendix, we assume that $\veps \le {\rm acos} \sqrt{\tfrac{2}{3}} \approx 35^{\circ}$ and that the state is reconstructed via linear inversion whenever possible, or otherwise via the maximum-likelihood (MLE) estimation technique.
 
 
 \subsubsection{Outcome data compatible with $\tau$, $\veps$}
 \label{App:Mixed1Qubit:OutcomeData}
 
We start by noting that for the purpose of state reconstruction, instead of the outcome probability distribution computed in Eq.~\eqref{Eq:Tomography:Outcomes}, we can just as well work with the vector $\vec{c}\in\mathbb{R}^3$ defined as follows:\footnote{There is a one-to-one correspondence between the components of $\vec{c}$ and the measured probability distribution $P(\pm1|k)$.}
 \begin{equation} \label{Eq:App:Correlators}
   \vec{c}\quad\text{ such that }\quad c_k=P(+1|k)-P(-1|k)=\tr(\tau N_k).
 \end{equation}
In the absence of misalignment error, i.e., when $N_k=M_k$ for all $k$, $\vec{c}$ is simply the vector of average values with respect to the Pauli matrices.

We now characterize the set $\mathcal{C}$ of $\vec{c}$ obtainable from the actual state $\tau = (\one + \vec{t} \cdot \vec{\sigma})/2$ and bounded misalignment error $\veps$. Using Eqs.~\eqref{Eq:Tomography:ActualState} and \eqref{Eq:App:Correlators}, we see that $\vec{c}$ and $\vec{t}$ are related through a linear transformation:
 \begin{equation}
   \label{Eq:App:Phi}
   \vec{c} = \tr \left[\left( \hat{n}_i \cdot \vec{\sigma} \right) \tau\right] =
   \underset{\equiv \Phi}{\underbrace{\left(\begin{array}{ccc}
           \hat{n}_1^{\left( x \right)} & \hat{n}_1^{\left( y \right)} &
           \hat{n}_1^{\left( z \right)}\\
           \hat{n}_2^{\left( x \right)} & \hat{n}_2^{\left( y \right)} &
           \hat{n}_2^{\left( z \right)}\\
           \hat{n}_3^{\left( x \right)} & \hat{n}_3^{\left( y \right)} &
           \hat{n}_3^{\left( z \right)}
         \end{array}\right)} \vec{t}},
 \end{equation}
 where $\Phi$ is a real matrix that encodes the actual measurement directions, $\hat{n}_k^{(x)}$ is the $x$-component of the unit vector $\hat{n}_k$; $\hat{n}_k^{(y)}$ and $\hat{n}_k^{(z)}$ are analogously defined. To parametrize $\mathcal{C}(\tau, \veps)$, we recall from Eq.~\eqref{Eq:eps} that the misalignment errors are bounded such that $\hat{n}_k$ satisfies 
\begin{equation}\label{Ineq:eps}
	\cos \veps \le \hat{n}_k \cdot \hat{m}_k.
\end{equation}
Additionally, we observe from Eq.~\eqref{Eq:App:Phi} that each component of the vector $\vec{c}$ can be written as $c_k=\hat{n}_k\cdot\vec{t}$.   Thus for given $\vec{t}$, each $c_k$ is constrained with an interval. More precisely, the set $\mathcal{C}(\tau, \veps)$ is a box whose boundaries are  specified by vectors saturating inequality~\eqref{Ineq:eps}.

 
 To simplify the computation, we now show that $\mathcal{C}(\tau,\veps)$ is contained inside a ball $\mathcal{B}(\vec{t},t \lambda)$ of radius $t \lambda$ centered at $\vec{t}$ (see Figure~\ref{fig_ball}), i.e., $\mathcal{B}(\vec{t},t \lambda) \supset \mathcal{C}(\tau,\veps)$, where
 \begin{equation}
   \label{Eq:App:Ball}
    \lambda \equiv 1 - \cos \veps + \sqrt{2} \sin \veps,
 \end{equation}
 and $\veps \lesssim 35^\circ$ ensures $0 \le \lambda \le 1$.
 \begin{figure}[h!]
   \includegraphics{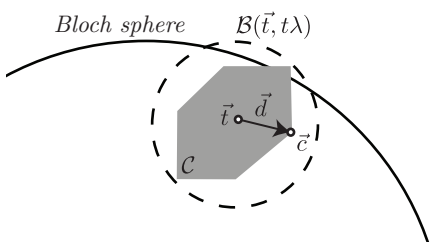}
   \caption{
     \label{fig_relaxation}
     \label{fig_ball}
     A schematic view of the set $\mathcal{C}$ of  vector $\vec{c}$ compatible with an actual state $\tau(\vec{t})$ measured
     with maximal misalignment error $\veps$; the set $\mathcal{C}$ is shaded
     in this 2D projection. In our proof, we work with its enclosing sphere
     $\mathcal{B}(\vec{t},t \lambda)$ whose boundary is marked with the dashed line.
   }
 \end{figure}

\begin{proof} 
To prove $\mathcal{C}(\tau,\veps)\subset \mathcal{B}(\vec{t},t \lambda)$, let us  take $\hat{n}_k$ that saturate inequality~\eqref{Ineq:eps} and decompose the matrix $\Phi$ in Eq.~\eqref{Eq:App:Phi} as
 \begin{equation}\label{Eq:Decomposition}
   \Phi - \mathbbm{1} = \left( \cos \varepsilon - 1 \right) \mathbbm{1}_3 +
   \sin \varepsilon \underset{\Psi}{\underbrace{\left(\begin{array}{ccc}
           0 & c_1 & s_1\\
           s_2 & 0 & c_2\\
           c_3 & s_3 & 0
         \end{array}\right)}},
 \end{equation}
 where $s_i = \sin \varphi_i$ and $c_i = \cos \varphi_i$. This allows us to determine the size of the enclosing ball $\B$ via the norm of the vector (see Figure~\ref{fig_ball}):
 \begin{equation}\label{Eq:Diff}
   \vec{d} = \vec{c} - \vec{t} = \left( \Phi - \mathbbm{1} \right) \vec{t} = t
   \left( \Phi - \mathbbm{1} \right) \, \hat{t} ,
   \quad t \equiv \norm{\vec{t}}.
 \end{equation}
 Since the maximal spectral radius of the matrix $\Psi$ in Eq.~\eqref{Eq:Decomposition}  is $\sqrt{2}$, the spectral radius of $\Phi - \mathbbm{1}$ is upper-bounded by $\lambda$ [as defined in Eq.~\eqref{Eq:App:Ball}]. Then the norm of the difference vector $\vec{d}$ is upper-bounded as follows:
 \begin{equation}
   d = \norm{\vec{d}} \le t \norm{\Phi - \mathbbm{1}}  \norm{\hat{t}} \le t\lambda
 \end{equation}
 which shows that a ball of radius $t \lambda$ centered at $\vec{t}$ indeed encloses all vectors $\vec{c}$ obtainable from $\tau$ assuming bounded misalignment error  $\veps$.
 \end{proof}
 
 \subsubsection{Reconstructed states compatible with $\vec{c}\in\mathcal{C}(\tau$, $\veps)$}

Whenever the vector $\vec{c}$ represents a legitimate Bloch vector, the reconstructed state $\rho$ follows immediately from Eq.~\eqref{Eq:Tomography:Reconstruction}:
 \begin{equation}
   \label{Eq:App:Reconstruction}
   \|\vec{c}\| \le 1 \quad \implies \quad \rho = \frac{\one + \vec{c} \cdot \vec{\sigma}}{2}.
 \end{equation}
This is true even if the state is reconstructed by the MLE technique. Whenever $\|\vec{c}\| > 1$, Eq.~\eqref{Eq:Tomography:Reconstruction} fails, but a physical state can still be reconstructed using the MLE technique by finding a quantum state $\rho$ with Bloch vector $\vec{r}$ that maximizes the likelihood function $\mathcal{L}(\vec{r})$,\footnote{This  can be achieved, for example, by using the iterative algorithm described in Refs.~\cite{Hradil:1997,MLE,Rehacek:2007}.} 
 where $\vec{r}$ is constrained by $\| \vec{r} \| \le 1$. 
 
By computing the Hessian of $\log \mathcal{L}(\vec{r})$, one can check that $\mathcal{L}$ is {\em strictly} concave in $\vec{r}$, with an (unconstrained) maximum at $\vec{r} = \vec{c}$. Therefore, for $||\vec{c}||>1$, the solution $\vec{r}\,^*$  that maximizes $\mathcal{L}$ must  lie on the boundary on the Bloch sphere. In particular the line segment joining $\vec{r}\,^*$ and $\vec{c}$ cannot cross the Bloch sphere; otherwise, it would contradict the strict concavity of $\mathcal{L}$.
 We can thus restrict our attention to $\vec{r}$ that lies on the surface of a spherical cap delimited by tangents of the Bloch sphere passing through $\vec{c}$. In Figure~\ref{Fig:App:MLE}, we plot the vector $\vec{r}$ that has the maximal angular deviation from $\vec{t}$ while being perpendicular to the tangential plane containing $\vec{c}$. Note, however, that depending on the actual functional form of $\mathcal{L}$, the state reconstructed from MLE may have an angular deviation that is less than $\theta$.
 \begin{figure}
   \includegraphics{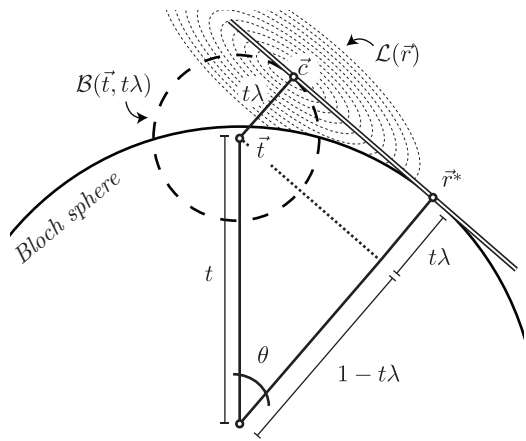}
   \caption{
     \label{Fig:App:MLE}
     Reconstructed state $\vec{r}\,^*$ using MLE, compatible with outcome data $\vec{c}$ outside the Bloch sphere. We show in the text that $\vec{c}$ is contained inside a ball of radius $t \lambda$ centered at $\vec{t}$ (dashed line). The reconstructed Bloch vector $\vec{r}\,^*$ lies on the surface of the Bloch sphere. Moreover,  since  the function $\mathcal{L}(\vec{r})$ is strictly concave --- schematic representation of the contour lines of $\mathcal{L}(\vec{r})$ are plotted using dots --- $\vec{r}\,^*$ is constrained on a spherical cap such that the line segment (drawn with a double edge) between $\vec{r}\,^*$ and $\vec{c}$ does not cross the boundary of the Bloch sphere. The maximal angle $\theta$ between vectors $\vec{t}$ and $\vec{r}\,^*$ is then obtained (see text) when the mentioned line segment is tangent to both balls.
   }
 \end{figure}
 
What is the maximal $\theta$ allowed? Since the outcome vector $\vec{c}$ is contained within $\mathcal{B}(\vec{t}, t \lambda)$, the angle $\theta$ is maximal when the line passing through $\vec{c}$ and $\vec{r}$ is tangent to both the Bloch sphere and $\mathcal{B}(\vec{t}, t \lambda)$ (as shown in Figure~\ref{Fig:App:MLE}). Standard trigonometry then gives:
 \begin{equation}
   \label{Eq:App:CosTheta}
   \cos \theta = \frac{1-t \lambda}{t} \implies \frac{\vec{r} \cdot \vec{t}}{t} \ge  \frac{1-t \lambda}{t}.
 \end{equation}

 \subsubsection{Worst-case fidelity}
 We now recall from Ref.~\cite{Mendonca:2008} that the fidelity function for single qubit states can be written as
 \begin{equation}
   \label{Eq:Fidelity}
   \F(\tau(\vec{t}),\rho(\vec{r}))=\frac{1}{2}\left[1+\vec{t}\cdot\vec{r}+\sqrt{(1-t^2)(1-r^2)}\right],
 \end{equation}
 which is concave for any given $\tau$, and has maximal value $1$ when $\rho = \tau$.
 
 To compute the worst-case fidelity $\F(\vec{t},\vec{r})$, we need to consider two separate cases. For $\| \vec{c} \| > 1$, the state $\rho$ is pure and using Eq.~\eqref{Eq:App:CosTheta}, we get:
 \begin{equation}
   \label{Eq:App:MLECase1}
   \mathcal{F} = \frac{1 + \vec{r} \cdot \vec{t}}{2}  \ge 1 - \frac{t \lambda}{2} \ge 1 - \frac{\lambda}{2}.
 \end{equation}
 
For $\| \vec{c} \| \le 1$, the reconstructed Bloch vector $\vec{r}=\vec{c}$ lies within the Bloch sphere. To compare the minimal fidelity attainable in this case with Eq.~\eqref{Eq:Tomography:FidelityLoss}, we shall consider the intersection of $\B(\vec{t}, \lambda t)$ and the Bloch sphere --- a convex set which we shall denote by $\B'$. Note that $\B'$ is still a superset of the set of outcome vector $\vec{c}$ compatible with $\tau$ and bounded misalignment error $\veps$, hence,
\begin{equation}\label{Eq:BoundFid}
	\min_{\| \vec{c} \| \le 1} \F(\vec{t},\vec{c})\ge \min_{\vec{c'}\in\B'} \F(\vec{t},\vec{c'})
\end{equation}
 By the concavity of $\mathcal{L}$, the right-hand-side of Eq.~\eqref{Eq:BoundFid} is attained at the boundary of $\B'$. Here, we can distinguish two subcases, namely,  the minimizing $\vec{c'}\in\B'$  corresponds to a (1) pure state or (2) mixed state, cf. Figure~\ref{Fig:App:MLE}.  
 
Now, we  remind that the fidelity function Eq.~\eqref{Eq:Fidelity} depends only on the inner product between the vectors as well as their magnitude. Without loss of generality, we can thus write these vectors in the 2-dimensional subspace spanned by $\vec{t}$ and $\vec{c'}$. For example, in the first case, we may write
\begin{equation}
	\vec{t}=(t,0)\quad\text{and}\quad\vec{c'}= \left(\cos\alpha, \sin\alpha\right),\,\,\alpha\in[-\alpha_c,\alpha_c],
\end{equation}
whilst in the second case, we may write
\begin{equation}\label{Eq:Interior}
	\vec{t}=(t,0)\quad\text{and}\quad\vec{c'}= \left( t - \lambda t\cos\theta,  \lambda t \sin\theta \right)
	,\,\,\theta\in[-\theta_c,\theta_c],
\end{equation}
where $\alpha_c=\text{acos}\left(\frac{1+t^2-\lambda^2t^2}{2t}\right)$ and $\theta_c=\text{acos}\left(\frac{t^2+\lambda^2t^2-1}{2\lambda t^2}\right)$.
Minimizing the fidelity for these two subcases, one finds that
\begin{equation}\label{Eq:BoundFid}
	\min_{\vec{c'}\in\B'} \F(\vec{t},\vec{c'})\ge 1-\frac{\lambda}{2}.
\end{equation}
Likewise, in the scenario where the ball $\B(\vec{t},t\lambda)$ in entirely contained within the Bloch sphere, i.e., when $\B(\vec{t},t\lambda)=\B'$, one can apply a parametrization similar to Eq.~\eqref{Eq:Interior} to show that the minimal fidelity also satisfies Eq.~\eqref{Eq:BoundFid}.
All in all, we thus see that the minimal fidelity when $\tau$ is a mixed state is always greater than or equal to worst-case fidelity for single-qubit pure state, i.e., $f(\veps)=1-\lambda / 2$. Thus Eq.~\eqref{Eq:Tomography:FidelityLoss} is a valid bound on the minimal fidelity for arbitrary single-qubit state.

\subsection{The MLE reconstruction of a product state remains product}
\label{App:Proof}

It can be shown that the product nature of a multipartite product quantum state $\tau=\bigotimes_{j=1}^n \tau^{(j)}$  is {\em preserved} during the MLE reconstruction procedure~\cite{MLE} even if some local errors --- such as the systematic errors envisaged in the main text --- incurred in the description of the positive-operator-valued-measure (POVM) elements. 

Specifically, let us denote by $\tau$ the actual state that undergoes the state tomography experiment, and let  $P(a_1, a_2|k_1, k_2)$ be the conditional probability of obtaining measurement outcomes $a_1, a_2$ for the choice of measurement settings $k_1, k_2$, and let  $\Pi^{(1)}_{a_1, k_1}, \Pi^{(2)}_{a_2, k_2}$ be the {\em actual} local POVM element used to generate these measurement statistics, i.e.:
\begin{equation}
\label{Eq:App:Product:P}
P(a_1, a_2|k_1, k_2)=\tr\,\left[\tau\, \left ( \Pi^{(1)}_{a_1, k_1} \otimes \Pi^{(2)}_{a_2, k_2} \right ) \right].
\end{equation}

We prove below the claimed proposition for the bipartite scenario. Its generalization to the $n$-partite scenario is evident.
\begin{proposition}\label{Thm:Product}
For measurement statistics gathered by performing local measurements $ \Pi^{(1)}_{a_1, k_1} \otimes \Pi^{(2)}_{a_2, k_2}$ on a bipartite product quantum state $\tau=\tau^{(1)} \otimes \tau^{(2)}$, any algorithm maximizing the likelihood function given in Ref.~\cite{MLE} reconstructs a product multipartite state, even if the algorithm employs a different set of POVM, say, 
$\left\{\tilde{\Pi}^{(j)}_{a_j, k_j}\right\}\neq \left\{\Pi^{(j)}_{a_j, k_j}\right\}$ for some $j$ and $k_j$.
\end{proposition}

\begin{proof}
First, let us recall that the quantum state $\rho$ maximizing the likelihood function given in Ref.~\cite{MLE} satisfies the following equation:
\begin{equation}
\label{Eq:App:Product:Rrho}
\rho = R \, \rho,
\end{equation}
where the operator $R$ encodes the maximization problem. Our proof is valid independently of the particular technique used to solve the maximum likelihood problem encoded in Eq.~\eqref{Eq:App:Product:Rrho}, for example, the iterative~\cite{MLE} and diluted algorithms~\cite{Rehacek:2007}. The operator $R$ is derived in Ref.~\cite{MLE} as:

\begin{equation}
\label{Eq:App:Product:R}
R = \sum_{a_1 a_2 k_1 k_2}
\frac{P(a_1,a_2|k_1,k_2)}
{\tr [ \rho\, (\tilde{\Pi}^{(1)}_{a_1,k_1} \otimes \tilde{\Pi}^{(2)}_{a_2,k_2}) ]}\, \tilde{\Pi}^{(1)}_{a_1,k_1} \otimes \tilde{\Pi}^{(2)}_{a_2,k_2}.
\end{equation}

We now introduce the ansatz $\check{\rho} = \check{\rho}^{(1)} \otimes \check{\rho}^{(2)}$, and substitute it into Eq.~\eqref{Eq:App:Product:Rrho}. First, the product structure $P(a_1,a_2|k_1,k_2)=P(a_1|k_1)P(a_2|k_2)$ follows by replacing $\tau = \tau^{(1)} \otimes \tau^{(2)}$ in Eq.~\eqref{Eq:App:Product:P}. Then, we exhibit the product structure of $R$ by introducing $\check{\rho}$ into Eq.~\eqref{Eq:App:Product:R}, giving $R=R^{(1)} \otimes R^{(2)}$ with:
\begin{equation}
\label{Eq:App:Product:R12}
R^{(j)} = \sum_{a_j k_j} \frac{P(a_j|k_j)}
{\tr [ \rho^{(j)} \tilde{\Pi}^{(j)}_{a_j,k_j} ]} \tilde{\Pi}^{(j)}_{a_j,k_j}.
\end{equation}

We may now rewrite Eq.~\eqref{Eq:App:Product:Rrho} as:
\begin{equation}
\check{\rho}^{(1)} \otimes \check{\rho}^{(2)} = \left ( R^{(1)} \otimes R^{(2)} \right ) \left (  \check{\rho}^{(1)} \otimes \check{\rho}^{(2)} \right ),
\end{equation}
and see immediately that original equation can be decomposed as analogous equations for the individual subsystems. Since the solution to Eq.~\eqref{Eq:App:Product:R} is unique~\cite{Rehacek:2007}, we thus see that the resulting reconstructed state must be  $\rho = \check{\rho} = \check{\rho}^{(1)} \otimes \check{\rho}^{(2)}$, where $\check{\rho}^{(j)}$ is the solution of the the single-qubit MLE equation $\check{\rho}^{(j)}=R^{(j)}  \check{\rho}^{(j)}$.
\end{proof}
 
 \subsection{Tomography of two-qubit pure states}
 \label{App:TwoQubit}
 
\subsubsection{Paramterization for low entanglement\\ (concurrence $\sin 2\alpha \lesssim 0.56$)}
\label{Sec:Tomography:TwoQubits:SmallConcurrence}

For $\alpha$ in this domain ($0\le \alpha \lesssim 0.29$), our numerical results can be very well approximated\footnote{This gives a $0.1\%$ relative error in terms of the loss of fidelity.}  by considering:
\begin{equation}\label{Eq:2qubit}
\left | \psi_\alpha \right > = e^{i\frac{2\pi}{3}} \cos \alpha \left | \psi_s^+ \right >
\left | \psi_s^+ \right > + \sin \alpha \left | \psi_s^- \right >
\left | \psi_s^- \right >,
\end{equation}
in conjunction with the {\em actual} measurement bases given in Eq.~\eqref{Eq:StandardBasis} for both qubits, 
where $\ket{\psi_s^\pm}$ were defined in Eq.~\eqref{Eq:psi_s_pm}. Note that this parametrization, in particular, recovers the optimal solution found for $\alpha=0$.

 \subsubsection{Robustness against misalignment error for low entanglement (concurrence $\sin 2\alpha \lesssim 0.56$)}
\label{App:RobSysErrLowEnt}

Here, we provide some intuition on the observation that pure two-qubit states become increasingly robust against the kind of systematic error that we consider as entanglement (parameterized by $\sin\,2\alpha$) increases within the aforementioned domain (cf. Figure~\ref{fig_twoqubit_tomography}). To this end, let us first rewrite $\ket{\psi_\alpha}$ as a density matrix, i.e., $\tau=\proj{\psi_\alpha}$ and remind that it can be decomposed in the basis of Pauli matrices:
\begin{equation}\label{Eq:HSDecomposition}
\tau = \frac{1}{4}\Big(\one\otimes\one + \vec{t}_1 \cdot
    \vec{\sigma} \otimes \mathbbm{1} + \mathbbm{1} \otimes \vec{t}_2
    \cdot \vec{\sigma} + \sum_{i j}
    T_{i j} \sigma_i \otimes \sigma_j\Big),
\end{equation}
where $\vec{t}_{1,2}$ are the Bloch vectors of the reduced density matrices, and $T$ is a $3\times3$ matrix that is responsible for the correlations between the two qubits. 

Note that for small $\alpha$, $\tau$ is weakly entangled and is close to a product state in the following
sense: from Eq.~\eqref{Eq:2qubit} and Eq.~\eqref{Eq:psi_s_pm}, if we decompose $T_{i j}$ as $T = \vec{t}_1 \vec{t}_2^T +
\tilde{T}$, we see that $\vec{t}_j = \cos 2 \alpha\,\,\,
\hat{s}$. Then the product term $\vec{t}_1 \, \vec{t}_2^T$ has spectral norm $\cos^2 2 \alpha$, whereas the spectral norm of $\tilde{T}$ is $|| \tilde{T} ||_2 = \sin 2 \alpha$. 
Rewriting $\tau$ as:
\begin{equation}
\tau=
\frac{
\left ( \one + \vec{t}_1 \cdot \vec{\sigma} \right )
}{2}
\otimes
\frac{
\left ( \one + \vec{t}_2 \cdot \vec{\sigma} \right )
}{2}
+
\frac{
\sum_{i j}
    \tilde{T}_{i j} \sigma_i \otimes \sigma_j
}{4},
\label{Eq:Tomography:Decomposition}
\end{equation}
it becomes clear that the contribution of $\tau$ when computing the
fidelity is mainly due to the first term in the sum. If we now
approximate $\tau$ by keeping only the first (product) term in the
sum, we approximate the state reconstruction by solving a set of
linear equations analogous to that given in
Eq.~\eqref{Eq:Tomography:Reconstruction}. The reconstructed state $\rho$ is also a product, as proven in Appendix~\ref{App:Proof}, and in particular the reconstructed Bloch vector is $(\cos \varepsilon + \sqrt{2} \sin \varepsilon)\vec{t}_j$, which has a norm proportional to $\cos2\alpha$. It thus follows that the susceptibility is more pronounced for smaller $\alpha$.

We can also understand this more formally by analyzing the loss of 
fidelity for small $\veps = \tfrac{\pi}{180}$:
\begin{equation}
\label{Eq:LossOfFidelity}
L(\alpha) = 1 - \F(\rho, \tau) = 1 - \tr \left(\rho\, \tau\right)
\end{equation}
on the optimal $\rho$ and pure state $\tau$ obtained from our numerical analysis \footnote{The loss of fidelity $L(\alpha)$ is proportional to the susceptibility $\mathcal{S}(\alpha)$ defined in Eq.~\eqref{Eq:Tomography:Susceptibility} at first order for small $\veps$ : $\mathcal{S}(\alpha) \approx \tfrac{L(\alpha)}{\veps}$.}.
Specifically, we  write $\tau = \tau_s +
\tilde{\tau}$ with $\tau_s = \tr_2 \tau \otimes \tr_1
\tau$, and similarly $\rho = \rho_s +
\tilde{\rho}$ and decompose $L(\alpha)$ as:
\begin{equation}
	L(\alpha) =\tr [\rho \, (\rho - \tau)] = \tr [\rho_s \,
	(\rho_s - \tau_s) ]+ \text{ terms in $\tilde{\rho},
	  \tilde{\tau}$}.
\label{eq_product_decomposition}
\end{equation}
The quantity $L(\alpha)$, as well as the two terms in the right-hand-side of Eq.~\eqref{eq_product_decomposition} are plotted individually as a function of $\alpha$ in Figure~\ref{fig_contributions}. Clearly, from the figure, we can see that the major  ($\approx89\%$) contribution to $L(\alpha)$ comes from the marginal term $\tr [ \rho_s \, (\rho_s - \tau_s)]$. 
\begin{figure}[h]
  \includegraphics{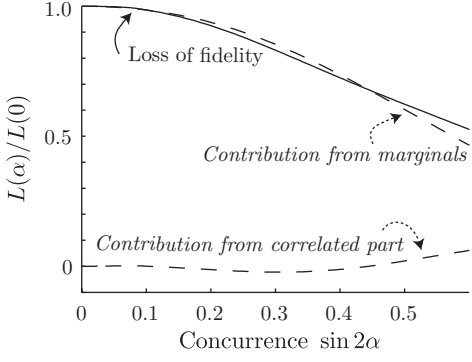}
  \caption{\label{fig_contributions}Loss of fidelity in the (numerically determined) worst-case
  scenario for low-entanglement states, for small $\veps=\tfrac{\pi}{180}$. The normalized loss of
fidelity is plotted as a solid line, contribution from the marginal terms $\tr [\rho_s \, (\rho_s - \tau_s) ]$ and the rest are
plotted as dashed lines.}
\end{figure}

\subsubsection{Parametrization for high entanglement\\ (concurrence $\sin 2\alpha \gtrsim 0.87$) }
\label{App:HighConcurrence}

 In this domain of $\alpha$, the correlation term in Eq.~\eqref{Eq:HSDecomposition} becomes the dominating term in $\tau$ and the parametrization given in Appendix~\ref{Sec:Tomography:TwoQubits:SmallConcurrence} no longer serves as a good approximation. Instead, a better parametrization\footnote{This gives a $0.2\%$ relative error in terms of the loss of fidelity. } to the optimal $\ket{\psi_\alpha}$ and $\hat{n}_k^{(j)}$ that we found in our optimization is:
 \begin{equation}
   \left | \psi_\alpha \right > = \cos \alpha \left | \psi_\theta^+ \right >
   \left | \psi_\theta^+ \right > + \sin \alpha \left | \psi_\theta^- \right >
   \left | \psi_\theta^- \right >,
 \end{equation}
 where $\left | \psi_\theta^\pm \right >$ are
 eigenvectors of  
 \begin{equation}
   \sigma_{\theta} = \frac{\sin \theta}{\sqrt{2}} \left( \sigma_x+ \sigma_z \right) + \cos \theta \, \sigma_y
 \end{equation}  
 with $\pm1$ eigenvalues and  $\theta \approx0.9961$. The phase of $\left | \psi_\theta^\pm \right >$ is such that
 $\left < 0 \middle | \psi_\theta^+ \right > = e^{i \phi } c^+$ and $\left < 0 \middle | \psi_\theta^- \right > = c^-$, where $c^\pm$ are some real numbers and $\phi \approx 0.4980$. 
 
 The state $\ket{\psi_\alpha}$ is then measured along the actual measurements axes
 \begin{equation}
   \label{Eq:App:HighConcurrence:Axes}
   \hat{n}_1^{(j)} = \left(\begin{array}{c}
       c_\varepsilon \\
       s_\veps\,c_\gamma\\
       s_\veps\,s_\gamma\\
     \end{array}\right),
   \hat{n}_2^{(j)} = \left(\begin{array}{c}
       -\tfrac{s_\varepsilon}{\sqrt{2}}\\
       c_\varepsilon \\
       -\tfrac{s_\varepsilon}{\sqrt{2}}
     \end{array}\right),
   \hat{n}_3^{(j)} = \left(\begin{array}{c}
       s_\veps\,s_\gamma\\
       s_\veps\,c_\gamma\\
       c_\varepsilon 
     \end{array}\right),
 \end{equation}
 where $s_{\varepsilon} = \sin \varepsilon$, $c_{\gamma} = \cos \gamma$, $s_{\gamma} = \sin \gamma$ and $\gamma \approx 2.7946$.

 \subsection{Tomography of two-qubit pure states with correlated misalignment error}
 \label{App:CorrelatedError}
 
\begin{figure}[h!]
  \includegraphics{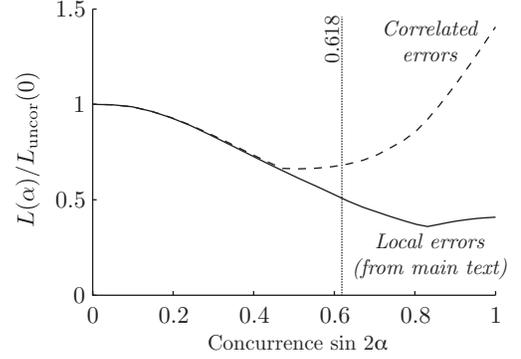}
  \caption{
    \label{Fig:App:CorrelatedError}
    Comparison of the loss of fidelity, as defined in Eq.~\eqref{Eq:LossOfFidelity} and for $\veps=\tfrac{\pi}{180}$, in the (numerically determined) worst-case scenario with correlated and local systematic errors. This graph is comparable at first order to the one in Figure~\ref{fig_twoqubit_tomography}.}
\end{figure}

The misalignment systematic errors considered in the main text are {\em{local}} in the sense that the measurement settings on the second party deviate from the ideal ones such that $\hat{m}_{k_2}^{\left( 2 \right)}$ is replaced by $\hat{n}_{k_2}^{\left( 2 \right)}$; this deviation does not depend on the measurement being done on the first party.
In contrast, let us consider now the case where the misalignment errors are {\em{correlated}} between different parties. This happens, e.g. in some ion traps where measurements on one physical system also changes the state of a neighboring system; or more commonly when pairs of settings $(k_1,k_2)$ are measured sequentially by realigning the measurement devices for each pair\cite{Tomography0}. Then the intended $k_2$-th measurement $\hat{m}^{(2)}_{k_2}$ may deviate in a different way when measuring the pair $(k_1,k_2)$ or $(k'_1,k_2)$. We thus replace the ideal measurement settings $\hat{m}_{k_j}^{\left( j \right)}$ performed on the $j$-th party by the actual settings $\hat{n}_{k_1,k_2}^{\left( j \right)}$, which now depends on $k_1$ and $k_2$.

We now compare the loss of fidelity in the local and correlated cases, fixing the maximal misalignment error at 1 degree ($\varepsilon = \tfrac{\pi}{180}$ rad) and numerically determine the worst-case fidelity in both cases. Numerically, we observe that both correlated and local systematic errors give the same fidelity drop when $\tau$ is a product state: $L \left( \alpha = 0 \right) \simeq 0.025$. Our results are shown in Figure~\ref{Fig:App:CorrelatedError}.

For small $\alpha$, we have seen in Appendix~\ref{Sec:Tomography:TwoQubits:SmallConcurrence} that the state $\tau$ has negligible correlated content $\left\| \tilde{T} \right\|_2$. In the {\em correlated} scenario, the major contribution to the loss of fidelity turns out to come also from the marginal part $\tau_s$. Moreover, the correlated errors seem to have a stronger effect on the correlated content, whose contribution to the loss of fidelity is negligible for small $\alpha$. In this regime, we thus observe the same behavior for the two scenarios.

When the correlated content of $\tau$ starts to dominate the marginal content (e.g., when $\left\|\tilde{T} \right\|_2 \gtrsim \left\| \vec{t}_1 \, \vec{t}_2^T \right\|_2$, which takes place for concurrence $\sin 2 \alpha \gtrsim 0.618$), the effect of correlated systematic errors becomes dominant.  In fact, the maximally entangled state ($\sin 2 \alpha = 1$) has the maximal loss of fidelity in this scenario.

\end{document}